\documentclass[a4paper,onecolumn,11pt,accepted=2020-09-10]{quantumarticle}
\pdfoutput=1

\usepackage[utf8]{inputenc}
\usepackage[english]{babel}
\usepackage[T1]{fontenc}
\usepackage[left=2.5cm, right=2.5cm, top=2.5cm, bottom=2.5cm]{geometry}
\usepackage{amssymb,amsmath}
\usepackage{amsthm}
\usepackage{xcolor}
\usepackage{algorithm}
\usepackage{graphicx}
\usepackage[noend]{algpseudocode}
\usepackage{array}

\usepackage[final=true,colorlinks = true,allcolors = {blue}]{hyperref}

\newtheorem{theorem}{Theorem}
\newtheorem{proposition}{Proposition}
\newtheorem{lemma}{Lemma}

\newtheorem{definition}{Definition}

\newcommand{\V}{\mathcal{V}}
\newcommand{\E}{\mathcal{E}}
\renewcommand{\S}{\mathcal{S}}
\newcommand{\A}{\mathcal{A}}

\DeclareMathOperator{\cost}{cost}
\renewcommand{\Pr}{\mathbb{P}}
\DeclareMathOperator{\Exp}{\mathbb{E}}
\DeclareMathOperator{\QS}{\mathrm{QS}}
\newcommand{\tO}{\widetilde{O}}
\newcommand{\tOmega}{\widetilde{\Omega}}
\newcommand{\Id}{\mathbb{I}}

\def\ket#1{|{#1}\rangle}
\def\braket#1{\langle{#1}\rangle}

\begin{document}

\title{Expansion Testing using Quantum Fast-Forwarding and Seed Sets}
\author{Simon Apers}
\affiliation{CWI (The Netherlands) and ULB (Belgium)}
\email{smgapers@gmail.com}
\orcid{0000-0003-3823-6804}
\thanks{Most of the work was done while part of the CWI-Inria International Lab.}

\maketitle

\begin{abstract}
Expansion testing aims to decide whether an $n$-node graph has expansion at least $\Phi$, or is far from any such graph.
We propose a quantum expansion tester with complexity $\widetilde{O}(n^{1/3}\Phi^{-1})$.
This accelerates the $\widetilde{O}(n^{1/2}\Phi^{-2})$ classical tester by Goldreich and Ron [Algorithmica '02], and combines the $\widetilde{O}(n^{1/3}\Phi^{-2})$ and $\widetilde{O}(n^{1/2}\Phi^{-1})$ quantum speedups by Ambainis, Childs and Liu [RANDOM '11] and Apers and Sarlette [QIC '19], respectively.
The latter approach builds on a quantum fast-forwarding scheme, which we improve upon by initially growing a seed set in the graph.
To grow this seed set we use a so-called evolving set process from the graph clustering literature, which allows to grow an appropriately local seed set.
\end{abstract}

\section{Introduction and Summary}

The (vertex) expansion of a graph is a measure for how well connected the graph is.
For an undirected graph $G = (\V,\E)$, with $|\V| = n$ and $|\E| = m$, it is defined as
\[
\Phi(G)
= \min_{\S \subset \V: |\S|\leq n/2} \frac{|\partial \S|}{|\S|},
\]
where $\partial \S$ is the set of nodes in $\V\backslash\S$ that have an edge going to $\S$.
See \cite{leighton1999multicommodity} for a discussion on the relevance of expansion for a range of graph approximation algorithms, and \cite{hoory2006expander} for a survey on expander graphs and their applications.
Since exactly determining $\Phi(G)$ is an NP-hard problem \cite{louis2013complexity}, we consider the relaxed objective of \textit{testing} the expansion.
Goldreich and Ron \cite{goldreich2002property,goldreich2011testing} initially studied this problem in the bounded-degree model, where they proposed the following question: given query acces to~$G$, does it have expansion at least some~$\Phi$, or is it far from any graph having expansion~$\tOmega(\Phi^2)$?
In this model, given graphs $G$ and $G'$ with degree bound~$d$, $G$ is $\epsilon$-far from $G'$ if at least $\epsilon nd$ edges have to be added or removed from $G$ to obtain $G'$.
They proved an $\Omega(n^{1/2})$ lower bound on the query complexity of this problem, and proposed an elegant tester based on random walk collision counting with complexity\footnote{In this section we hide polynomial dependencies on $\log n$, the degree bound $d_M$ of the graph, and the distance parameter $\epsilon$.
In the rest of the paper, $\tO$ simply hides any poly-logarithmic dependencies.}
\[
\tO(n^{1/2}\Phi^{-2}).
\]
In rough strokes, the algorithm picks a uniformly random node, and counts collisions between~$\tO(n^{1/2})$ independent random walks of length~$\tO(\Phi^{-2})$ all starting from this node.
If the graph is far from being an expander, then the random walk will get stuck in certain low-expansion subsets, leading to an increased number of collisions.
The graph is hence rejected if the number of collisions exceeds some constant.

\begin{figure}[htb]
\centering
\includegraphics[width=.7\textwidth]{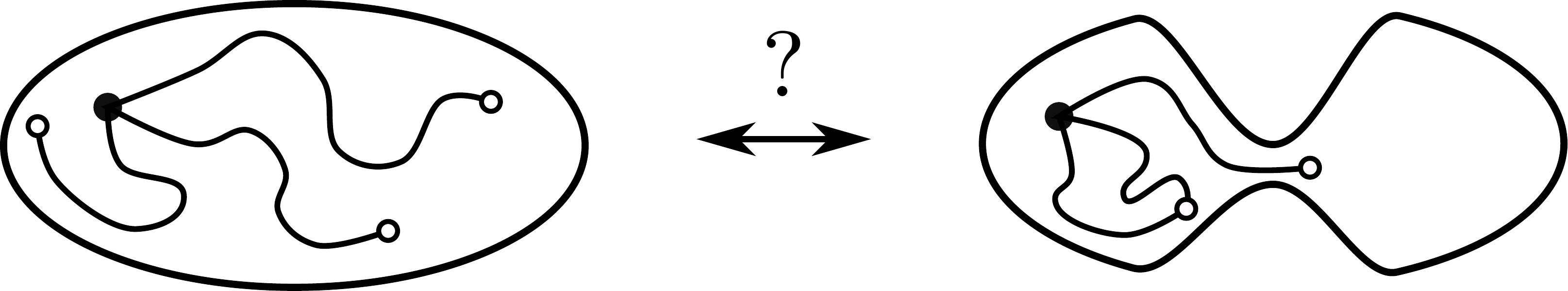}
\caption{The GR tester counts collisions between independent random walks starting from some seed node. Low expansion of the graph results in an increased number of collisions.} \label{fig:GR-tester}
\end{figure}

Goldreich and Ron had to base the correctness of their tester on a certain unproven combinatorial conjecture.
However, in later works by Czumaj and Sohler \cite{czumaj2010testing}, Kale and Seshadhri \cite{kale2011expansion} and Nachmias and Shapira \cite{nachmias2010testing} the correctness was unconditionally established.
The ideas underlying this tester and its analysis were more recently extended towards testing the $k$-clusterability of a graph \cite{czumaj2015testing,chiplunkar2018testing}, which is a multipartite generalization of the expansion testing problem.

In this work we consider the expansion testing problem in the quantum setting, where we allow to perform queries in superposition.
We refer to the nice survey by Montanaro and de Wolf \cite{montanaro2016survey} for a general overview of quantum property testing.
Ambainis, Childs and Liu \cite{ambainis2011quantum} were the first to describe a quantum algorithm for expansion testing.
The gist of their algorithm is to combine an appropriate derandomization of the GR tester with Ambainis' quantum algorithm for element distinctness \cite{ambainis2007quantum}.
The latter allows to count collisions among the set of $\tO(n^{1/2})$ random walk endpoints using only $\tO(n^{1/3})$ quantum queries.
The improved complexity of their quantum expansion tester is
\[
\tO(n^{1/3} \Phi^{-2}).
\]
In addition they proved an $\tOmega(n^{1/4})$ lower bound on the quantum query complexity.

In later work of the current author together with Sarlette \cite{apers2019quantum}, as well as in the current work, a very different approach is taken.
Quantum walks, which form the quantum counterpart of random walks, are used to explore the graph.
Rather than picking random neighbors, a quantum walk explores a graph through quantum queries to its neighborhood.
In particular, this allows to create a ``quantum sample'' that appropriately encodes the random walk distribution.
As we detail in Section \ref{sec:QFF-tester} below, we can then use standard tools from quantum algorithms to estimate the random walk collision probability.
In \cite{apers2019quantum} we introduced a new quantum walk technique called ``quantum fast-forwarding'' (QFF) that allows to approximately prepare these quantum samples in the square root of the random walk runtime.
This yielded a new quantum expansion tester with complexity
\[
\tO(n^{1/2} \Phi^{-1}),
\]
quadratically improving the dependency of the GR tester on $\Phi$, which corresponds to the random walk runtime.
Up to this work, this left the problem of quantum expansion testing with two different testers with a complementary speedup.
In this work, however, we present a new quantum tester which closes this gap.
Essentially we improve the QFF tester from \cite{apers2019quantum} by initially doing some classical work in the graph: from the initial node $v$, we first grow a local node subset or ``seed set'' of size $n^{1/3}$.
In earlier work by the author \cite{apers2019qsampling} it was already shown that such seed sets allow to more efficiently create quantum samples, essentially by improving the projection of the initial state on the final quantum sample.
Indeed, starting from this seed set, rather than directly from $v$, we can run the QFF tester with an improved complexity $\tO(n^{1/3} \Phi^{-1})$.
To prove correctness of the tester, we must ensure that if the initial node $v$ is inside some low-expansion set, then the seed set largely remains inside that set.
Thereto we borrow a so-called ``evolving set process'' from the local graph clustering literature \cite{andersen2016almost}, allowing to grow such a set in complexity $\tO(n^{1/3} \Phi^{-1})$.
This allows to prove our main result:
\begin{theorem}
There exists a quantum expansion tester with complexity $\tO(n^{1/3} \Phi^{-1})$.
\end{theorem}
\noindent
The resulting speedup combines the quantum speedups of \cite{ambainis2011quantum} and \cite{apers2019quantum}.
To summarize, we gather the different algorithms and approaches in Table \ref{tab:qet}.

\begin{table}
\centering
\def\arraystretch{1.5}%
\begin{tabular}{| m{5.7cm} | m{3cm} | m{5cm} |}
\hline
Goldreich and Ron \cite{goldreich2002property} & $\tO(n^{1/2} \Phi^{-2})$ (conj.) & RW collision counting \\
\hline
Czumaj and Sohler \cite{czumaj2010testing}, \newline
Kale and Seshadhri \cite{kale2011expansion}, \newline
Nachmias and Shapira \cite{nachmias2010testing} & $\tO(n^{1/2} \Phi^{-2})$ & prove GR conjecture \\
\hline
Ambainis, Childs and Liu \cite{ambainis2011quantum} & $\tO(n^{1/3} \Phi^{-2})$ (q) & quantum element distinctness \\
\hline
Apers and Sarlette \cite{apers2019quantum} & $\tO(n^{1/2} \Phi^{-1})$ (q) & QFF \\
\hline
this work & $\tO(n^{1/3} \Phi^{-1})$ (q) & QFF and seed sets \\
\hline
\end{tabular}
\caption{Complexity for expansion testing. (q) denotes quantum complexity.}
\label{tab:qet}
\end{table}

\subsection{QFF Tester} \label{sec:QFF-tester}
Our tester builds on the QFF tester from \cite{apers2019quantum}, hence we describe this tester first.
Let $P$ denote the random walk (RW) transition matrix, and $P^t\ket{v}$ the $t$-step RW probability distribution\footnote{We use the ket-notation $\ket{v}$ to simply denote the indicator vector on node $v$.} starting from a node $v$.
The tester builds on the observation that the squared 2-norm $\|P^t\ket{v}\|^2$ exactly equals the collision probability of a pair of random walks:
\[
\|P^t\ket{v}\|^2
= \sum_{u \in \V} P^t(u,v)^2
= \sum_{u \in \V} \Pr(X_t=u \,|\, X_0=v)^2,
\]
where we let $X_t$ denote the random walk position at time step $t$.
Hence, we can estimate the collision probability between the $t$-step RW endpoints simply by estimating $\|P^t\ket{v}\|^2$.
This we can do rather straightforwardly using quantum algorithms, in particular making use of quantum walks (QWs).
Starting from an initial node $v$ of the graph, a QW allows to generate the quantum sample
\[
\ket{\psi_t}
= P^t\ket{v} + \ket{\Gamma},
\]
which encodes the RW probability distribution $P^t\ket{v}$ as one of its component, with $\ket{\Gamma}$ some auxiliary garbage component that is orthogonal to the RW component, and which we will not care about.
Given the ability to generate such quantum samples, we can then use a standard quantum routine called \textit{quantum amplitude estimation} to estimate the norm $\|P^t\ket{v}\|^2$ of the RW component.
Now, similarly to the GR tester, if we set $t \in O(\Phi^{-2})$ then we can reject the graph if $\|P^t\ket{v}\|^2$, and hence the collision probability, is larger than some threshold.

The amplitude estimation routine requires $\tO(\|P^t\ket{v}\|^{-1})$ quantum samples, so that the complexity of this quantum expansion tester is $\tO(\|P^t\ket{v}\|^{-1} \QS_t)$, where $\QS_t$ denotes the quantum complexity of creating the quantum sample $\ket{\psi_t}$.
If the graph is regular\footnote{Later on we adapt the graph to ensure this.} then $P$ has a uniform stationary distribution, i.e., the vector $\ket{\pi} = n^{-1/2} \sum_{u\in\V} \ket{u}$ is the unique eigenvalue-1 eigenvector of $P$.
This allows to bound
\begin{equation} \label{eq:norm-bnd}
\|P^t\ket{v}\|
\geq |\braket{\pi|v}|
= n^{-1/2},
\end{equation}
so that $\tO(\|P^t\ket{v}\|^{-1} \QS_t) \in \tO(n^{1/2} \QS_t)$.
In order to bound $\QS_t$, we can use an existing QW approach by Watrous \cite{watrous2001quantum} which gives $\QS_t \in O(t)$.
Since we choose $t \in \tO(\Phi^{-2})$, this yields a complexity $\tO(n^{1/2} \Phi^{-2})$, thus giving no speedup with respect to the GR tester.
In \cite{apers2019quantum} however, we introduced a more involved QW technique called \textit{quantum fast-forwarding} (QFF).
Building on a Chebyshev truncation of the $P^t$ operator, this technique allows to quadratically improve the complexity to $\QS_t \in O(t^{1/2})$, resulting in a complexity
\[
\tO(n^{1/2} \Phi^{-1}).
\]
Given that the GR tester has complexity $\tO(n^{1/2} \Phi^{-2})$, this yields a complementary speedup to the $\tO(n^{1/3} \Phi^{-2})$ expansion tester in \cite{ambainis2011quantum}.
Whereas their speedup follows from a quantum routine for accelerating the collision counting procedure, the speedup in the QFF tester follows from accelerating the random walk runtime.

\subsection{QFF Tester with Seed Sets}
In this paper we refine the QFF tester, improving its complexity to $\tO(n^{1/3} \Phi^{-1})$.
We improve its suboptimal $n^{1/2}$-dependency by initially constructing or ``growing'' a seed set around the initial node, from which we then run the QFF tester.
This idea is derived from earlier work of the author \cite{apers2019qsampling}, where seed sets are used to create a superposition over the edges of a graph, leading to a similar speedup.
The main insight is derived from the bound in \eqref{eq:norm-bnd}, showing that the suboptimal $n^{1/2}$-dependency stems from a small projection of the initial state $\ket{v}$ onto the uniform superposition~$\ket{\pi}$.
Growing a seed set allows to improve this dependency: if we grow a set $\S \subseteq \V$ from $v$, and we use the quantum superposition $\ket{\S} = |\S|^{-1/2} \sum_{u\in\S} \ket{u}$ as an initial state, this bound becomes
\[
\|P^t\ket{S}\|
\geq |\braket{\pi|\S}|
= |\S|^{1/2} n^{-1/2}.
\]
This suggests the following new tester: (i) pick a uniformly random node $v$, (ii) grow a seed set $\S$ from $v$ of appropriate size, and (iii) create $\tO(n^{1/2}|\S|^{-1/2})$ QW samples $\ket{\psi_t} = P^t\ket{\S} + \ket{\Gamma}$, allowing to estimate $\|P^t\ket{\S}\|$.
Assuming that the construction of $\S$ requires $|\S|$ queries, and momentarily ignoring the $\Phi$-dependency, this tester has a combined complexity of $\tO(|\S| + \|P^t\ket{S}\|^{-1}) \in \tO(|\S| + n^{1/2}|\S|^{-1/2})$.
If we choose $|\S| = n^{1/3}$, this becomes $\tO(n^{1/3})$ as we aimed for.

Using a similar reasoning as before, we again wish to reject the graph if our estimate is larger than some threshold.
Indeed, as depicted in Figure \ref{fig:QET}, if the seed set is localized in some low-expansion set, then the 2-norm $\|P^t\ket{S}\|$ will be larger than when the graph has no low-expansion sets.
The difficulty however is to ensure that the seed set $\S$, when grown from some initial node $v$ in a low-expansion set, effectively remains inside that set.
If this is not the case, then a RW from $\S$ will no longer be stuck in the low-expansion set, thus no longer giving rise to an increased 2-norm.
As a consequence, we cannot simply use a breadth-first search from $v$, as we did in \cite{apers2019qsampling}: a BFS might exit a low-expansion set more easily than a random walk.
Luckily, however, the problem of locally exploring a low-expansion set (or ``cluster'') turns out to be well-studied under the name ``local graph clustering'' \cite{spielman2013local,andersen2006local,andersen2016almost,orecchia2012approximating}.

\begin{figure}[htb]
\centering
\includegraphics[width=.8\textwidth]{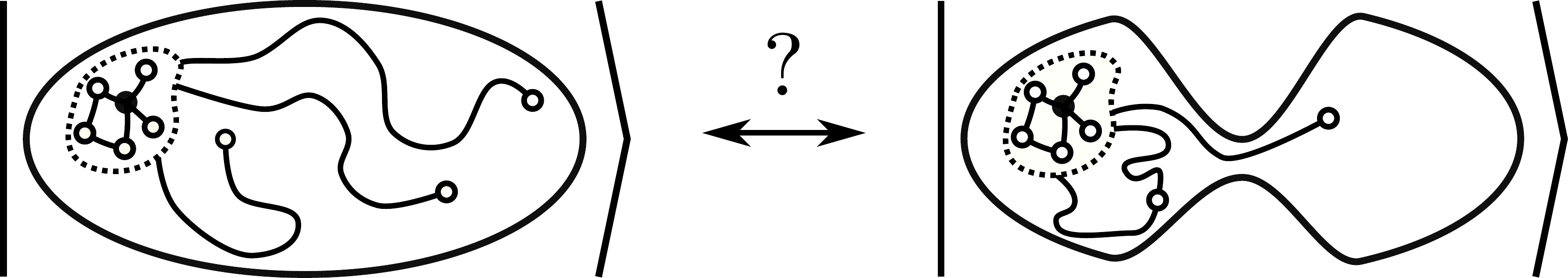}
\caption{The new tester classically grows an appropriately local seed set around the initial node.
From this set a quantum sample can be generated more efficiently.
We use an evolving set process to ensure that the seed set mostly remains inside the initial low-expansion set.} \label{fig:QET}
\end{figure}

In particular, we can use a so-called ``evolving set process'' (ESP) as used by Andersen, Oveis-Gharan, Peres and Trevisan in \cite{andersen2016almost}.
An ESP is a Markov chain on subsets of the nodes, which evolves by expanding or contracting its boundary based on the RW behavior on the graph.
Given an initial node $v$ inside a low-expansion set, they simulate an ESP to explicitly retrieve this set.
Since we are interested in growing a potentially much smaller seed set $\S$ inside the cluster, we slightly adapt their algorithm, leading to the following result.
The algorithm either returns a low-expansion set, allowing to immediately reject the graph, or it returns an appropriate seed set.
\begin{proposition}
Fix a parameter $M \geq 0$.
Given a random node $v$ from a set $\S'$ of expansion $\tO(\Phi^2)$, we can use an ESP to return a set $\S$ such that with constant probability either $\Phi(\S) < \Phi$ or $|\S| \geq M$ and $|\S \cap \S'|/|\S| \in \Omega(1)$.
The complexity of generating this set is $\tO(M \Phi^{-1})$.
\end{proposition}

Building on this tool, we can now sketch our new quantum tester, summarized in Algorithm~\ref{alg:qget-intro}.
Since the ESP process requires $\tO(n^{1/3} \Phi^{-1})$ steps by the above proposition, and estimating $\|P^t\ket{\S}\|$ requires $\tO(\|P^t\ket{\S}\| \Phi^{-1}) \in \tO(n^{1/3} \Phi^{-1})$ steps, we retrieve the promised tester complexity $\tO(n^{1/3} \Phi^{-1})$.

\begin{algorithm}
\caption{Quantum Expansion Tester} \label{alg:qget-intro}
\normalsize
\begin{algorithmic}[1]
\State
select a uniformly random starting node $v$
\State
grow a seed set $\S$ from $v$ using an ESP
\State
\textbf{if} $\Phi(\S) < \Phi$ \textbf{then} \texttt{reject}
\State
use quantum amplitude estimation to estimate $\| P^t \ket{\S} \|$ for $t \in \tO(\Phi^{-1})$
\State
\textbf{if} $\| P^t \ket{\S} \|$ too large \textbf{then} \texttt{reject} \textbf{else} \texttt{accept}
\end{algorithmic}
\end{algorithm}

\subsection{Open Questions}
We finish this section by discussing some open questions related to this work.
\begin{itemize}
\item
In \cite{apers2019qsampling} a breadth-first search is used to grow a seed set $\S$, requiring a number of steps $\tO(|\S|)$.
In the current work we use a more refined ESP algorithm to grow $\S$, which in particular ensures that the set remains inside some low-expansion subset (say with expansion at most $\Phi$).
This procedure however requires an increased number of steps $\tO(|\S|\Phi^{-1})$.
We leave it as an open question whether such an appropriate set can be grown in $\tO(|\S|)$ steps.
The complexity of the tester then becomes $\tO(|\S| + n^{1/2}|\S|^{-1/2}\Phi^{-1})$.
Setting $|\S| = n^{1/3}\Phi^{-2/3}$ this would lead to an improved complexity $\tO(n^{1/3}\Phi^{-2/3})$.
\item
The use of an ESP for the expansion testing problem could also be useful for improving the $\Phi$-dependency of the classical GR tester.
If we could for instance grow a pair of seed sets, both of size $n^{1/2}$, that behave to some extent as ``random subsets'' of a local cluster, then we could simply count collisions between these sets, thus avoiding the use of random walks.
A higher number of collisions would then again signal a low expansion of the graph.
Ideally an ESP-like procedure would allow to grow these sets in $\tO(n^{1/2}\Phi^{-1})$ steps, improving on the $\tO(n^{1/2}\Phi^{-2})$ complexity of the GR tester.
\item
Clusterability testing, as recently studied in \cite{czumaj2015testing,chiplunkar2018testing}, uses very similar techniques to the GR expansion tester.
It seems feasible that we can use the techniques from this paper to similarly improve on these testers.
\item
Goldreich and Ron \cite{goldreich2002property} proved a classical lower bound $\Omega(n^{1/2})$ for expansion testing, suggesting that their tester has an optimal dependency on $n$.
In the quantum setting however, the only known lower bound is $\Omega(n^{1/4})$ as proven by Ambainis, Childs and Liu \cite{ambainis2011quantum}, thus leaving a large gap to all current quantum testers, which have a $\tO(n^{1/3})$-dependency.
While our work does not provide any new insights towards closing this gap, we do feel that this is an interesting question to resolve.
\end{itemize}

\section{Preliminaries} \label{sec:prelim}

In this section we formalize the computational model, the query model and the definition of an expansion tester.
We also describe some necessary random walk properties, and define the notion of a quantum walk.

\subsection{Complexity, Computational Model and QRAM} \label{sec:comp-model}

The \textit{quantum query complexity} of an algorithm simply denotes the number of quantum queries that the algorithm makes to the input (see Section \ref{sec:query-model} for details).

In contrast, the actual runtime or \textit{complexity} of the algorithm is defined with respect to a computational model.
We specify our computation model as a system that
\begin{enumerate}
\item
can run quantum routines on $O(\log n)$ qubits, where $n$ is the number of nodes in the input graph,
\item
can make quantum queries to the input (see Section \ref{sec:query-model}), and
\item
has access to a quantum-read/classical-write RAM (QRAM) of $\tO(n^{1/3})$ classical bits.
A single QRAM operation corresponds to either (i) classically writing a bit to the QRAM, or (ii) making a quantum query to bits stored in the QRAM.
\end{enumerate}
By the \textit{complexity} of an algorithm we denote the total number of (i) queries, (ii) elementary classical and quantum operations (gates), and (iii) QRAM operations.
By definition, the complexity of an algorithm forms an upper bound on the query complexity of an algorithm, irrespective of the computational model or QRAM assumptions.

Let $\S \subset \V$ be a subset of nodes.
As in \cite{apers2019qsampling}, we can use the QRAM to efficiently create, and reflect around, the quantum state $\ket{\S} = |\S|^{-1/2} \sum_{v \in \S} \ket{v}$.
To this end, we rely on the following theorem by Kerenidis and Prakash \cite{kerenidis2016quantum}.
\begin{theorem}[{\cite[Theorem 15]{kerenidis2016quantum}}]
Assume that we are given a subset $\S \subset \V$.
With a one-off preprocessing cost of $\tO(|\S|)$ QRAM operations, we can repeatedly (i) create the quantum state $\ket{\S}$, and (ii) reflect around the quantum state $\ket{\S}$, both using $\tO(1)$ QRAM operations per repetition.
\end{theorem}

We make two final comments on our QRAM assumption, as it is often subject of debate.
First, such an assumption is native to any quantum algorithm that makes use of quantum-accessible classical input or memory.
We note that the actual qubits that a QRAM query acts on (i.e., the superposition of queried addresses plus the target register) is only logarithmic in the number of stored bits.
In that sense our memory requirement is weaker than for instance that of Ambainis’s element distinctness algorithm \cite{ambainis2007quantum} and its application in the quantum expansion tester of \cite{ambainis2011quantum}, which effectively require a polynomial-sized memory of qubits.
Second, as already pointed out above, the complexity is an upper bound on the query complexity, irrespective of the computational model or QRAM assumptions.
In that sense, our work improves on the query complexity of the former works \cite{ambainis2011quantum,apers2019quantum}, also without these assumptions.

\subsection{Query Model and Property Testing} \label{sec:query-model}
We are given query access to some undirected graph $G = (\V,\E)$, with node set $\V = [n]$ and edge set $\E$.
We denote $|\E| = m$.
For any $v \in \V$, we let $d(v)$ denote the degree of $v$, the maximum degree $d_M = \max_{v\in\V} d(v)$, and $d(\S) = \sum_{v\in\S} d(v)$ denotes the total degree of a set $\S \subseteq \V$.
We say that $G$ has degree bound $d$ if $d_M \leq d$.
In the context of property testing of bounded-degree graphs \cite{goldreich2002property}, the following queries are allowed:
\begin{itemize}
\item
\textit{uniform node query}: return uniformly random node $v \in \V$
\item
\textit{degree query}: given $v \in \V$, return degree $d(v)$
\item
\textit{neighbor query:} given $v \in \V$, $k \in [d(v)]$, return $k$-th neighbor of $v$
\end{itemize}
Throughout the paper we will assume that $G$ is regular.
If this is not the case, then we can always modify the graph to ensure this: to any node $i$ with degree $d(i)<d$ we add $d-d(i)$ parallel self loops.
This effectively renders the graph regular, ensuring that a random walk converges to the uniform distribution, as we will require later.
Notably this does not change the expansion of the graph.
Goldreich and Ron \cite{goldreich2011testing} achieve the same effect by modifying the random walk rather than the graph, but modifying the graph will prove more elegant for our purpose.

Since we wish to study quantum algorithms, we will allow to perform degree and neighbor queries in superposition.
To illustrate this, assume that a neighbor query, given $v\in\V$ and $k\in[d(v)]$, returns a node $u$.
Using quantum notation, this is described as a unitary transformation
\[
\ket{v}\ket{k}\ket{x}
\mapsto \ket{v}\ket{k}\ket{x+u},
\]
where $x$ is some arbitrary $\lceil\log n\rceil$-bit string, and ``$\displaystyle +$'' denotes addition modulo $\lceil\log n\rceil$.
We can now imagine the first register being in a superposition $d(v)^{-1/2} \sum_{k\in[d(v)]} \ket{v}\ket{k}\ket{x}$, so that the query operation now becomes
\begin{equation} \label{eq:quantum-query}
\frac{1}{\sqrt{d(v)}} \sum_{i\in[d(v)]} \ket{v}\ket{k}\ket{x}
\mapsto \frac{1}{\sqrt{d(v)}} \sum_{i\in[d(v)]} \ket{v}\ket{k}\ket{x+u^{(k)}},
\end{equation}
where we let $u^{(k)}$ denote the $k$-th neighbor of $v$.
We will call a single such query a ``quantum query''.
We refer the interested reader to the survey by Montanaro and de Wolf \cite{montanaro2016survey} for more details on the quantum query model.

We will follow the property testing model for bounded-degree graphs by Goldreich and Ron \cite{goldreich2002property}.
Given two $n$-node graphs $G=([n],\E)$ and $G'=([n],\E')$ with degree bound $d$, they define the relative distance between $G$ and $G'$ as the number of edges that needs to be added or removed to turn $G$ into $G'$, divided by the maximum number of edges $nd$.
This is equal to $|\E \triangle \E'|/(nd)$, with $\triangle$ the symmetric difference between $\E$ and $\E'$.
$G$ is then said to be $\epsilon$-far from $G'$ if $|\E \triangle \E'|/(nd) \geq \epsilon$.
When studying a certain property $P$ of graphs, $G$ is said to be ``$\epsilon$-far from having property $P$'' if $G$ is $\epsilon$-far from any graph $G'$ having property $P$.

\subsection{Expansion Testing}
We define the (vertex) expansion of a subset $\S \subset \V$ as $\Phi(\S) = |\partial \S|/|\S|$.
Here $\partial \S = \{u \in \S^c \mid \exists v\in\S \;\mathrm{s.t.}\; (u,v)\in\E\}$ is the set of nodes in $\S^c$ that have an edge going to $\S$.
The expansion of a graph $G$ is then defined as
\[
\Phi(G)
= \min_{\S \subset \V: |\S|\leq n/2} \Phi(\S).
\]
We consider the following definition of an expansion tester due to Czumaj and Sohler \cite{czumaj2010testing}.
\begin{definition}
An algorithm is a $(\Phi,\epsilon)$-expansion tester if there exists a constant $c>0$, possibly dependent on $d$, such that given parameters $n$, $d$, and query access to an $n$-node graph with degree bound $d$ it holds that
\begin{itemize}
\item
if the graph has expansion at least $\Phi$, then the algorithm outputs ``{\normalfont\texttt{accept}}'' with probability at least $2/3$,
\item
if the graph is $\epsilon$-far from any graph having expansion at least $c \Phi^2 \log^{-1}(dn)$, then the algorithm outputs ``{\normalfont\texttt{reject}}'' with probability at least $2/3$.
\end{itemize}
\end{definition}
\noindent
We note that this is a slightly more constrained definition than the one in e.g.~\cite{kale2011expansion,ambainis2011quantum,apers2019quantum}.
In these works the $\log$-factor in the reject case is actually left as an additional free parameter $\mu$, which is compensated in the runtime.
While our algorithm might also work in that more general setting, we believe that the corresponding technicalities would go beyond the scope and main ideas of this paper, and we leave it as a minor open question.
We also mention that in the traditional setting of property testing, the expression ``$c \Phi^2 \log^{-1}(dn)$'' in the second bullet should be replaced by ``$\Phi$''.
Although unproven, the relaxation in this definition seems necessary to allow for efficient (sublinear) testing using random walks.
This is a consequence of the fact that the expansion only characterizes the random walk mixing behavior up to a quadratic factor.
We stress that this quadratic gap is present in all works on expansion testing.

Apart from the vertex expansion, we also define the conductance.
When studying random walks, this is often a slightly more appropriate measure.
For a subset $\S \subset \V$ it is defined as $\phi(\S) = |\E(\S,\S^c)|/d(\S)$, where $\E(\S,\S^c) = \{(u,v)\in\E \mid u\in\S,\, v\in\S^c\}$ denotes the set of edges between $\S$ to $\S^c$.
The conductance of a graph $G$ with $m$ edges is then defined as $\phi(G) = \min_{\S \subset \V: d(\S)\leq m/2} \phi(\S)$.
If $G$ is $d$-regular, as we will assume throughout the paper, this simplifies to $\phi(G) = \min_{\S \subset \V: |\S|\leq n/2} |\E(\S,\S^c)|/(d|\S|)$.
Since $|\partial\S| \leq |\E(\S,\S^c)| \leq d|\partial\S|$, this allows to relate vertex expansion and conductance as follows:
\begin{equation} \label{eq:bnd-cond}
\Phi(\S)/d
\leq \phi(\S)
\leq \Phi(\S).
\end{equation}

\subsection{Random Walks}
We will consider lazy random walks (RWs), described by a Markov chain on the node set.
From any node the RW jumps with probability $1/2$ to any of its neighbors uniformly at random, and otherwise stands still.
If we let $P(u,v)$ denote the RW transition probability from node $v$ to node $u$, then $P(u,v) = 1/(2d(v))$ for $(v,u) \in E$, $P(u,v) = 1/2$ if $u = v$ and $P(u,v) = 0$ elsewhere.
If the underlying graph is connected, then the RW converges to a unique limit distribution in which every node has a probability proportional to its degree.
On a regular graph, this corresponds to a uniform distribution.

\subsubsection{Diffusion Core} \label{sec:diff-core}

Central to the study of expansion testers is the so-called ``diffusion core'' of a set $\S \subseteq \V$.
The diffusion core allows to lower bound the probability that a RW of given length stays entirely inside $\S$, as a function of its conductance $\phi(\S)$.
Let $\tau_v(\S^c)$ denote the escape time of $\S$ from $v$, i.e., the hitting time of a RW from $v\in\S$ to the complement $\S^c$.
We then define the diffusion core of $\S$ as follows:
\begin{definition} \label{def:diff-core}
For $\alpha,\beta>0$, the $(\alpha,\beta)$-diffusion core of $\S$ is defined as
\[
\S_{\alpha,\beta}
= \big\{ v \in \S \,\big\vert\, 
	\Pr(\tau_v(\S^c) > \alpha \phi(\S)^{-1}) \geq \beta \big\}.
\]
\end{definition}
\noindent
Throughout we define the ``canonical'' diffusion core $\S_d = \S_{1/40,3/4}$.
Using a reasoning similar to Spielman and Teng \cite{spielman2013local}, we can lower bound the size of the diffusion core.
\begin{lemma} \label{lem:diff-size}
\[
\frac{d(\S_{\alpha,\beta})}{d(\S)}
> 1 - \frac{\alpha}{2(1-\beta)}.
\]
\end{lemma}
\begin{proof}
Let $Y_v$ denote the event that $\tau_v(\S^c)>\alpha \phi(\S)^{-1}$, let $\pi$ denote the stationary distribution of the RW, and let $\pi_\S$ denote the distribution $\pi$ conditioned on being in the set $\S$: $\pi_\S(v) = \Id(v\in\S) \pi(v)/\pi(\S)$.
From \cite[Proposition 2.5]{spielman2013local} we know that $\Pr_{v\sim\pi_\S}(Y_v) \geq 1 - \alpha/2$.
For all $v \notin \S_{\alpha,\beta}$, it holds by definition that $\Pr(Y_v) < \beta$, so that we can bound
\[
\Pr_{v\sim\pi_\S}(Y_v)
= \sum_{v\in\S} \pi_\S(v) \Pr(Y_v)
< (1-\pi_\S(\S_{\alpha,\beta})) \beta + \pi_\S(\S_{\alpha,\beta}).
\]
Combined with the former inequality, and the fact that $\pi_\S(\S_{\alpha,\beta}) = d(\S_{\alpha,\beta})/d(\S)$, this proves the claimed statement.
\end{proof}

This lemma implies that $d(\S_d)/d(\S) > 19/20$.
As we will require this later, we also wish to prove something slightly stronger: there exists a subset $\S'$ of the diffusion core $\S_d$, from which we can bound the probability that a random walk stays inside the diffusion core, rather than only inside $\S$.

\begin{lemma} \label{lem:diff-core}
There exists a node subset $\S'$ of the diffusion core $\S_d$, with $d(\S') > d(\S)/3$, from which a $(120\phi(\S))^{-1}$-step RW remains inside $\S_d$ with probability at least $9/10$:
\[
\forall v \in \S': \quad
\Pr(\tau_v(\S_d^c) > (120 \phi(\S))^{-1}) \geq 9/10.
\]
\end{lemma}
\begin{proof}
In the following we use the shorthand $\phi = \phi(\S)$.
We can set $\S'$ equal to the $(1/30,39/40)$-diffusion core, $\S' = \S_{1/30,39/40}$.
Using Definition \ref{def:diff-core} we see that $\S' \subseteq \S_d \subseteq \S$.
From Lemma \ref{lem:diff-size} we know that $d(\S') > d(\S)/3$.

We will show that $\S'$ serves as a diffusion core for $\S_d$.
Thereto fix any $v \in \S'$ and let $\kappa$ denote the hitting time $\kappa = \tau_v(\S_d^c)$.
Then we define $t$ such that $\Pr(\kappa \leq t) > 1/10$.
Now let $Y$ be the event that $\kappa \leq t$ and $\tau_u(\S^c) \leq (40\phi)^{-1}$, with $u = X_\kappa$ a random variable corresponding to the node at which the RW hits $\S_d^c$.
Then we have that $Y \Rightarrow (\tau_v(\S^c) \leq t + (40\phi)^{-1})$.
Since $u \notin \S_d$, it holds that $\Pr(\tau_u(\S^c) \leq (40\phi)^{-1}) > 1/4$.
Combined with the assumption that $\Pr(\kappa \leq t) > 1/10$, this allows to bound $\Pr(Y) > (1/4)(1/10) = 1/40$, and therefore $\Pr(\tau_v(\S^c) \leq t + (2\phi)^{-1}) > 1/40$.
However, since $v \in \S'$ we also have that $\Pr(\tau_v(\S^c) > 1/(30\phi)) \geq 39/40$, or equivalently $\Pr(\tau_v(\S^c) \leq 1/(30\phi)) \leq 1/40$.
This gives a contradiction if $t + (40\phi)^{-1} \leq 1/(30\phi)$, or equivalently $t \leq 1/(120\phi)$.
For such $t$ the initial hypothesis $\Pr(\kappa \leq t) > 1/10$ must hence be false, and therefore it must hold that $\Pr(\kappa \leq (120\phi)^{-1}) \leq 1/10$ for all $v \in \S'$.
This proves the claimed statement.
\end{proof}

We will also use the following lemma, which in essence was already present in \cite{czumaj2010testing,nachmias2010testing,kale2011expansion}.
It argues that a graph which is $\epsilon$-far from having a certain expansion must have a large subset with low expansion.
\begin{lemma} \label{lem:reject-set}
Let $G$ be an undirected $n$-node graph with degree bound $d$ that is $\epsilon$-far from having expansion $\geq \beta$, with $\beta \leq 1/10$.
Then the following holds:
\begin{itemize}
\item
There exists a subset $\A \subset \V$, with $\epsilon n/12 \leq |\A| \leq (1+\epsilon)n/2$, such that $\Phi(\A) < r_d \, \beta$, with $r_d$ a constant dependent on $d$.
\item
For any $t \leq 1/(2r_d \beta)$ and distribution $v$ having a $\gamma$-overlap with the diffusion core of $\A$, with $\gamma > 2(1+\epsilon)/3$, it holds that
\[
\|P^t v\|^2
\geq \frac{1}{n}\left(1 + 4\left(\frac{3\gamma}{4} - \frac{1+\epsilon}{2}\right)^2\right).
\]
\end{itemize}
\end{lemma}
\begin{proof}
The first bullet is proven in \cite[Corollary 4.6]{czumaj2010testing}.
To prove the second bullet, we use the fact that for a general probability distribution $w$ it holds that
\[
\|w\|^2
\geq \frac{1}{n} \left(1 + \|w - u\|_1^2\right),
\]
with $u$ the uniform distribution.
This bound can be found in the proof of \cite[Lemma 4.3]{czumaj2010testing}.
To lower bound the right hand side, we will use that
\[
\| w - u \|_1
= 2 \max_{\S \subseteq \V} |w(\S) - u(\S)|
\geq 2 |w(\A) - u(\A)|.
\]
If $w = P^t v$, we can lower bound $(P^t v)(\A)$ since this represents the probability that a $t$-step RW from $v$ end in $\A$.
By definition of the diffusion core, we know that a $(t\leq 1/(40\Phi(\A)))$-step RW, starting anywhere in the diffusion core $\A_d$ of $\A$, remains inside $\A$ with probability at least $3/4$.
Since $v$ has a $\gamma$-overlap with $\A_d$, this proves that a $(t \leq 1/(40 r_d \beta))$-step RW from $v$ remains inside $\A$ with probability at least $3\gamma/4$.
This implies that $(P^t v)(\A) \geq 3\gamma/4$ and hence $\| P^t v - u \|_1 \geq 2 |3\gamma/4 - u(\A)|$.
By our assumption that $\gamma > 2(1+\epsilon)/3$, and since $u(\A) = |\A|/n \leq (1+\epsilon)/2$, this implies that $\| P^t v - u \|_1 \geq 2 (3\gamma/4 - (1+\epsilon)/2)$.
Combining this with the above inequality proves the claimed bound.
\end{proof}

\subsection{Quantum Walks}
Quantum walks (QWs) form an elegant quantum counterpart to random walks on graphs.
They similarly explore a graph in a local manner, by performing queries in superposition to the neighbors of certain nodes, as illustrated in \eqref{eq:quantum-query} in Section \ref{sec:query-model}.
In the following, let $P$ be a symmetric random walk transition matrix (as will be the case for us), and let $\S \subseteq \V$ be the initial seed set.
We denote by $\ket{\S} = |\S|^{-1/2} \sum_{v \in \S} \ket{v}$ the quantum state that is a uniform superposition over nodes in $\S$.
Starting from $\ket{S}$, QWs allow to create a quantum state or ``quantum sample'' of the form
\[
\ket{\psi_t}
= P^t\ket{\S} + \ket{\Gamma}.
\]
Here the first component forms a quantum encoding of the RW probability distribution started from a uniformly random node in $\S$.
The second component denotes some auxiliary garbage state in which we will not be interested.
In our earlier work on quantum expansion testing, we introduced a QW technique called ``quantum fast-forwarding'' (QFF) that allows to approximate the above quantum sample in the square root of the classical runtime.
The following lemmas follow from \cite{apers2019quantum}, recalling that $d_M$ denotes the maximum degree of the graph.
\begin{lemma}[QFF] \label{lem:QFF}
Starting from the state $\ket{\S}$, there exists a QW algorithm that outputs a state $\epsilon$-close to the quantum sample $\ket{\psi_t}$ in complexity $\tO(t^{1/2} d_M^{1/2} \log^{1/2}(1/\epsilon))$.
\end{lemma}
\begin{proof}
We first note that, starting from $\ket{\S}$, the state $\ket{\psi_t}$ can be $\epsilon$-approximated using a number of \textit{QW steps} that scales as $\tO(t^{1/2} \log^{1/2}(1/\epsilon))$ (hiding polylog-dependencies on $t$, $n$, $\epsilon$).
An explicit statement of this fact can be found as Theorem 7 in \cite{ambainis2019quadratic}.
Second we note that a single QW step can be implemented in complexity $O(d_M^{1/2})$, as is mentioned in \cite{apers2019quantum,apers2019qsampling}.
Combining both facts proves the lemma.
\end{proof}
For clarity of exposition, we will ignore the approximation error $\epsilon$ of QFF in the rest of the paper.
By linearity it suffices to set $\epsilon$ inverse polynomially small in $n$, and so the approximation error will only add a log-factor to the overall complexity (this in contrast to the approximation error in the lemma below, which in fact is important).

Given access to such quantum samples, we can use a standard quantum routine called ``quantum amplitude estimation'' to estimate $\|P^t\ket{\S}\|$.
This leads to the following lemma, which is an immediate corollary from \cite[Theorem 5]{apers2019quantum} and Lemma \ref{lem:QFF}.
\begin{lemma}[2-norm estimator] \label{lem:2-norm-est}
There exists a QW algorithm that, with probability at least $1-\delta$, outputs an estimate $a$ such that $\big| \|P^t \ket{v}\| - a \big| \leq \epsilon$.
The algorithm has complexity $\tO(t^{1/2} d_M^{1/2} \epsilon^{-1} \log\delta^{-1})$ and uses $\tO(\epsilon^{-1} \log\delta^{-1})$ reflections around $\ket{\S}$.
\end{lemma}

As discussed in Section \ref{sec:comp-model} (and equal to the approach in \cite{apers2019qsampling}), we can use a QRAM data structure to prepare and reflect around the quantum state $\ket{\S}$.
The effective complexity of these operations in our computation model is then $\tO(1)$ (i.e., polylogarithmic in $n$).

\section{Evolving Set Processes} \label{sec:ESP}

Evolving Set Processes (ESPs) have been used for analyzing the mixing time of Markov chains \cite{morris2005evolving}, and as an algorithmic tool for performing local graph clustering \cite{andersen2009finding,andersen2016almost}.
Derived from some original Markov chain over a node set $\V$, an ESP is a Markov chain over \textit{subsets} of the node set.
For our particular case we will assume that the original Markov chain corresponds to the (lazy) RW.
Given that the current state of the ESP is $\S \subseteq \V$, its next state is then determined by the following rule: draw a variable $U$ uniformly at random from the interval $[0,1]$, and set the next state
\[
\S'
= \big\{ v \in \V \,:\, P(\S,v) \geq U \big\}.
\]
Here $P(\S,v)$ denotes the probability that a single RW step from $v$ ends up in $\S$, and is given by $P(v,\S) = |\E(v,\S)|/(2d(v)) + \Id(v\in\S)/2$.
This gives rise to an ESP transition matrix $K: 2^\V \times 2^\V \to [0,1]$.
Notice that only states in the inner or outer boundary of $\S$ can be added or removed: $|\E(v,\S)| = d(v)$ and $\Id(v\in\S)=1$ if $y$ and all of its neighbors lie in $\S$, whereas $|\E(v,\S)| = \Id(v\in\S) = 0$ if $v$ nor any of its neighbors lie in $\S$.
This process has absorbing states $\S = \emptyset$ and $\S = \V$, both of which have no boundary.
For algorithmic purposes, it is desirable to prevent the ESP from being absorbed in the empty set.
To this end, the transition probabilities can be slightly altered:
\[
\hat{K}(\S,\S')
= \frac{d(\S')}{d(\S)} K(\S,\S').
\]
Clearly the transition probability to the empty set is now equal to zero.
$\hat{K}$ is again a stochastic transition matrix, and the resulting process is called the \textit{volume-biased} ESP (yet for brevity we will simply refer to it as the ESP).
We refer the reader to \cite{andersen2016almost,levin2017markov} for more details on the ESP and its volume-biased variant.

Starting inside some low-expansion set $\S$, ESPs are used as a means of locally constructing or exploring $\S$.
In our case, we only wish to retrieve a smaller subset, typically of size $|\S|^{1/3}$.
This subset however should be sufficiently localized ``inside'' $\S$, i.e., have a sufficient overlap with the smaller diffusion core of $\S$.
To this end we refine the ESP analysis: we use our Lemma \ref{lem:diff-core} to show that also the ESP will remain in the diffusion core with large probability.
The following Section \ref{sec:ESP-compl} introduces some useful properties of the ESP, and in Section \ref{sec:ESP-thm} we prove the main tool.

\subsection{ESP Complexity and Properties} \label{sec:ESP-compl}

As we wish to use an ESP as an algorithmic means, it is desirable to quantify the resources needed to simulate it.
Thereto we define the \textit{cost} of a sample path
\[
\cost(\S_0,\dots,\S_t)
= d(\S_0) + \sum_{i=1}^t \big( d(\S_i \triangle \S_{i-1}) + |\partial(\S_{i-1})| \big),
\]
with $d(\S)$ the total degree of a subset $\S$ and $\S \triangle \S'$ the symmetric difference between $\S$ and $\S'$.
We also define a stopping time $\tau(T,B,\theta)$ for the ESP:
\begin{definition}
The stopping time $\tau(T,B,\theta)$ is a random variable that equals the first time $\tau$ at which either $\phi(\S_\tau) \leq \theta$, $\tau = T$, or $\cost_\tau > B$.
\end{definition}
\noindent
The following theorem from \cite{andersen2016almost} bounds the complexity of sampling from the ESP with stopping rule $\tau(T,B,\theta)$.
\begin{theorem}[{\cite[Proposition 5.3]{andersen2016almost}}] \label{thm:ESP-cost}
There exists an algorithm that takes as input a node $v$, two integers $T,B \geq 0$ and $\theta \in [0,1]$.
Let $\S_0 = \{v\}$ and define the stopping time $\tau = \tau(T,B,\theta)$.
The algorithm generates a sample path $(\S_0,\dots,\S_\tau)$ of the ESP and outputs the last set $\S_\tau$.
The complexity of the algorithm is $O(B \log m)$.
\end{theorem}

\subsection{ESP for Growing Seed Set} \label{sec:ESP-thm}
Using known results on ESPs, combined with our new Lemma \ref{lem:diff-core}, we can derive the following theorem.
This constitutes the main tool that we will use to grow seed sets.
We defer the proof technicalities to Appendix \ref{app:proof-esp}.

\begin{theorem} \label{thm:ESP-seedset}
Fix a parameter $M \geq 0$.
Let $\S \subseteq \V$ be such that $\phi(\S) \leq \gamma^2/(2400\log m)$.
Let $\S' \subseteq \S$ be as defined in Lemma \ref{lem:diff-core}, and assume that $\S_0 = \{v\}$ for some $v \in \S'$.
Let $\S_\tau$ be the set returned by the ESP with stopping rule $\tau(T,B,\theta)$ and parameters $T = 20\theta^{-2} \log m$, $B = 25 M \sqrt{T \log m}$, and $\theta = \gamma$.
Then with probability at least $1/5$ we have that $d(\S_\tau \cap \S_d)/d(\S_\tau) \geq 3/4$, and either $\phi(\S_\tau) \leq \gamma$ or $d(\S_\tau) \geq M$.
The complexity of generating this set is $O(M \gamma^{-1} \log^2 m)$.
\end{theorem}

\section{Quantum Expansion Tester} \label{sec:QET}

We are now ready to construct our new quantum expansion tester, yielding the main contribution of this paper.

\begin{algorithm}[H]
\caption{Quantum Expansion Tester} \label{alg:qget}
\normalsize
\textbf{Input:} parameters $n$ and $d$; query access to an $n$-node graph $G$ with degree bound $d$; expansion parameter $\Phi$; promise parameter $\epsilon$ \\ 
\textbf{Do:}
\begin{algorithmic}[1]
\State
set parameters:\newline
$t = 16d^2 \Phi^{-2} \log n$, $\delta = \epsilon/1000$, $K = 200 / (\epsilon(1-\delta))$,\newline
$\theta = \Phi/(2d)$, $B = 800 \sqrt{5} \lceil n^{1/3} d \rceil d \Phi^{-1} \log m$, $T = 320 \Phi^{-2} d^2 \log m$
\State
\textbf{do $K$ times}
\State \indent
select a uniformly random starting node $v$
\State \indent
run ESP from $\S_0 = \{v\}$ with stopping rule $\tau(T,B,\theta)$, outputting $\S$
\State \indent
\textbf{if} $|\S| \leq n/2$ and $\Phi(\S) \leq \Phi/2$ \textbf{then} abort and output ``\texttt{reject}''
\State \indent
use 2-norm estimator to create estimate $a$ of $\| P^t \ket{\S} \|$\newline
\-\ \hspace{5mm} to precision $\epsilon' = \sqrt{|\S|/n} (1-\sqrt{1+1/256})/4$ with probability $1-\delta$
\State \indent
\textbf{if} $a > \sqrt{|\S| n^{-1}(1+n^{-1})} + \epsilon'$ \textbf{then} abort and output ``\texttt{reject}''
\end{algorithmic}
\textbf{Output:} if no ``\texttt{reject}'', output ``\texttt{accept}''
\end{algorithm}

\begin{theorem}[Quantum Expansion Tester] \label{thm:QGET}
If $d \geq 3$ and $\epsilon < 1/16$, then Algorithm \ref{alg:qget} is a $(\Phi,\epsilon)$ expansion tester for $c = 1/(2400 (2d)^2 r_d)$, with $r_d$ as in Lemma \ref{lem:reject-set}.
The complexity of the algorithm is bounded by $\tO(n^{1/3} \Phi^{-1} d^{3/2} \epsilon^{-1})$.
\end{theorem}
\begin{proof}
First we prove that if $\Phi(G) \geq \Phi$, then the algorithm accepts with probability at least $2/3$.
Thereto note that by definition of the vertex expansion, necessarily $\Phi(\S) \geq \Phi(G) \geq \Phi$ if $|\S| \leq n/2$.
Hence the algorithm cannot falsely reject in step 5.
To exclude rejection in step 7, we use the result in \cite[Proof of Theorem 2.1]{nachmias2010testing} showing that if $\Phi(G) \geq \Phi$ then for all nodes $v\in\V$ and time $t \geq 16d^2 \Phi^{-2} \log n$ it holds that $\|P^t \ket{v}\| \leq \sqrt{n^{-1}(1+n^{-1})}$.
This allows to bound $\|P^t \ket{\S}\| \leq |\S|^{-1/2} \sum_{s\in\S} \| P^t \ket{s} \| \leq \sqrt{|\S| n^{-1}(1+n^{-1})}$.
Using the 2-norm estimator from Lemma \ref{lem:2-norm-est}, the estimate $a$ will with probability $1-\delta$ be such that $a \leq \sqrt{|\S| n^{-1}(1+n^{-1})} + \epsilon'$.
Step 7 will therefore reject falsely only with probability at most $\delta$.
The total probability of a faulty rejection can then be bounded by $K\delta < 1/3$.
Since the algorithm accepts if it never rejects, it will correctly accept the graph with probability at least $2/3$.

Next we prove that if $G$ is $\epsilon$-far from having expansion $\geq c \Phi^2 \log^{-1} (dn)$, then the algorithm rejects with probability at least $2/3$.
Thereto we use Lemma \ref{lem:reject-set} from Section \ref{sec:diff-core}, which states that in this case there exist a ``bad'' subset $\A$, with $(1+\epsilon)n/2 \geq |\A| \geq \epsilon n/12$, such that
\begin{equation} \label{eq:bnd-phiA}
\Phi(\A)
< r_d c \Phi^2 \log^{-1} (dn)
= \frac{1}{2400} \bigg(\frac{\Phi}{2d}\bigg)^2 \frac{1}{\log(dn)}.
\end{equation}
From $\A$, we can define the diffusion core $\A_d$ and the subset $\A' \subseteq \A_d$ as in Lemma \ref{lem:diff-core}, which states that $d(\A') > d(\A)/3$ and hence $|\A'| \geq |\A|/3$.
The initial node $v$ will hence be in $\A'$ with probability at least $|\A|/(3n) \geq \epsilon/36$.

Conditioning on $v \in \A'$, we can analyze the ESP output set $\S$ using Theorem \ref{thm:ESP-seedset}.
We choose $M = \lceil n^{1/3} d\rceil$.
Using the bound \eqref{eq:bnd-phiA}, which by \eqref{eq:bnd-cond} implies the same upper bound for $\phi(\A)$, $\S$ will with probability at least $1/5$ be such that (i) $d(\S \cap \A_d) \geq 3d(\S)/4$, and therefore $|\S \cap \A_d| \geq 3|\S|/4$, and (ii) either $\phi(\S) \leq \Phi/(2d)$ or $d(\S) \geq M$.
If $\phi(\S) \leq \Phi/(2d)$ and $|\S| \leq n/2$, then we have a proof that $G$ has vertex expansion $\Phi(G) \leq \Phi/2$, and hence we reject the graph in step 5.
To see this, simply note that $\phi(\S) \leq \Phi/(2d)$ implies that $\Phi(\S) \leq \Phi/2$ (again by \eqref{eq:bnd-cond}).
In any other case, we know that $d(\S) \geq M$ and hence $|\S| \geq M/d$.
Given such a set $\S$, consider the uniform distribution $\pi_\S$ over $\S$.
Since $|\S \cap \A_d| \geq 3|\S|/4$, we know that $\pi_\S$ has a $3/4$-overlap with $\A_d$.
By the second bullet of Lemma \ref{lem:reject-set} this implies that for all $t \leq \log(dn)/(2r_d c \Phi^2) = 1200 (2d)^2 \log(dn) \Phi^{-2}$,
\[
\|P^t \pi_\S\|
\geq \sqrt{\frac{1}{n} \left(1 + 4\left(\frac{3\gamma}{4} - \frac{1+\epsilon}{2}\right)^2 \right)}
\geq \sqrt{\frac{1}{n} \left(1 + \frac{1}{256}\right)}.
\]
using that $\gamma = 3/4$ and $\epsilon \leq 1/16$.
By Lemma \ref{lem:2-norm-est}, the estimate $a$ will then with probability $1-\delta$ be such that $a \geq \sqrt{|\S| n^{-1}} \sqrt{1 + 1/256} - \epsilon'$.
If $\epsilon' \leq \sqrt{|\S| n^{-1}} (\sqrt{1+1/256}-1)/4$, then this is strictly larger than $\sqrt{|\S| n^{-1}(1+n^{-1})} + \epsilon'$ for sufficiently large $n$, allowing to correctly reject the graph with probability at least $1-\delta$.

Now we can bound the total probability of correctly rejecting the graph in a single iteration.
Thereto we multiply the probability that $v \in \A'$ ($\geq \epsilon/36$), the ESP process succeeds ($\geq 1/5$) and the 2-norm estimator succeeds ($\geq 1-\delta$), yielding a total rejection probability of at least $p = \epsilon(1-\delta)/180$.
The total probability of correctly rejecting at least once in $K$ iterations is therefore at least $1-(1-p)^K$.
Using the elementary inequality $(1-p)^{1/p} < 1/e$ for any $0 < p \leq 1$, we can lower bound the rejection probability as
\[
1 - (1-p)^K
> 1 - \left(\frac{1}{e}\right)^{Kp}
\geq \frac{2}{3},
\]
provided that $K \geq \ln(3)/p = 180 \ln(3) / (\epsilon(1-\delta))$, which is ensured by our choice of $K$.
This concludes the proof that Algorithm \ref{alg:qget} is a $(\Phi,\epsilon)$-expansion tester.

Towards bounding the complexity of the algorithm, we consider a single iteration of the for-loop.
By Theorem \ref{thm:ESP-seedset}, the complexity of simulating the ESP in step 4 is $O(M \theta^{-1} \log^2(dn))$, which is $O(n^{1/3} d^2 \Phi^{-1} \log^2(dn))$.
By Lemma \ref{lem:2-norm-est}, the $(\epsilon',\delta)$ 2-norm estimator in step 6 has complexity
\[
\tO(t^{1/2} d^{1/2} \epsilon'^{-1} \log\delta^{-1})
\in \tO(n^{1/3} d^{3/2} \Phi^{-1} \log\epsilon^{-1}).
\]
Since we iterate the for-loop $K\in O(\epsilon^{-1})$ times, this gives the claimed complexity.
\end{proof}

\section*{Acknowledgements}
This work greatly benefited from discussions and comments by Alain Sarlette, Anthony Leverrier, Ronald de Wolf and Andr\'e Chailloux, as well as from comments and suggestions by multiple anonymous referees.
Most of the work was done while part of the CWI-Inria International Lab.
We acknowledge support from the QuantERA ERA-NET Cofund in Quantum Technologies implemented within the European Union's Horizon 2020 Programme (QuantAlgo project) and from the Belgian Fonds de la Recherche Scientifique - FNRS under grant no R.50.05.18.F (QuantAlgo).

\bibliographystyle{alpha2}
\bibliography{biblio}

\appendix

\section{Proof of ESP Algorithm} \label{app:proof-esp}

In this appendix we provide the proof of Theorem \ref{thm:ESP-seedset}.
We make use of several known properties that can be attributed to the output set:
\begin{itemize}
\item
\textit{Size}:
The cost of simulating the process is related to the size of the output set.
\begin{lemma}[{\cite[Theorem 5.4]{andersen2016almost}}] \label{lem:size}
For any starting set $\S_0$ and any stopping time $\tau$ that is upper bounded by $T$, it holds that
\[
\Exp[\cost_\tau/d(\S_\tau)]
\leq 1 + 4 \sqrt{T\log m}.
\]
\end{lemma}
\item
\textit{Overlap}:
If a random walk has a high probability of staying inside a certain set, then with high probability the ESP will also largely remain inside that set.
\begin{lemma}[{\cite[Lemma 4.3]{andersen2016almost}}] \label{lem:overlap}
Consider any set $\S \subseteq \V$, a starting set $\S_0 = \{v\}$ for some $v\in\S$, and an integer $T \geq 0$.
Then the following holds for all $\beta > 0$:
\[
\Pr \big( {\textstyle\min_{t\leq T}}\, d(\S_t \cap \S)/d(\S_t)
\geq 1 - \beta \, \Pr(\tau_v(\S^c) \leq T) \big)
\geq 1 - 1/\beta.
\]
\end{lemma}
\item
\textit{Conductance}:
After $T$ steps, the ESP encounters with high probability a set of conductance $\tO(T^{-1/2})$.
\begin{lemma}[{\cite[Corollary 1]{andersen2009finding}}] \label{lem:cond}
Fix any integer $T$, and let $\theta_T = \sqrt{4T^{-1} \log m}$.
For any starting set $\S_0$ and constant $c \geq 0$, it holds that
\[
\Pr \big( {\textstyle\min_{t<T}}\; \phi(\S_t) \leq \sqrt{c} \, \theta_T \big)
\geq 1 - 1/c.
\]
\end{lemma}
\end{itemize}

\noindent
Using these lemmas, combined with our Lemma \ref{lem:diff-core}, we can prove the theorem below.
It corresponds to Theorem \ref{thm:ESP-seedset} when setting the parameters $\alpha=5$ and $\beta = 5/2$.

\begin{theorem}
Fix constants $\alpha,\beta > 0$ and a parameter $M \geq 0$.
Let $\S \subseteq \V$ be such that
\[
\phi(\S)
\leq \frac{1}{480\alpha} \frac{\gamma^2}{\log m}.
\]
Let $\S' \subseteq \S$ be as defined in Lemma \ref{lem:diff-core}, and assume that $\S_0 = \{v\}$ for some $v \in \S'$.
Let $\S_\tau$ be the set returned by the ESP with stopping rule $\tau(T,B,\theta)$ and parameters $T = 4\alpha\theta^{-2} \log m$, $B = 5 \alpha M \sqrt{T \log m}$, and $\theta = \gamma$.
Then with probability at least $1 - 2\alpha^{-1} - \beta^{-1}$ we have that
\[
\frac{d(\S_\tau \cap \S_d)}{d(\S_\tau)} \geq 1 - \beta/10,
\]
and either
\[
\phi(\S_\tau) \leq \gamma \quad \text{ or } \quad
d(\S_\tau) \geq M.
\]
The complexity of generating this set is $O(M \gamma^{-1} \log^2 m)$.
\end{theorem}
\begin{proof}
Let $X$ denote the event that $d(\S_\tau \cup \S_d)/d(\S_\tau) \geq 1-\beta/10$, and $Y$ the event that $\phi(\S_\tau) \leq \gamma$ or $d(\S_\tau) \geq M$.
We will show that $\Pr(X) \geq 1 - \beta^{-1}$ and $\Pr(Y) \geq 1 - 2\alpha^{-1}$, so that by the union bound we find that $\Pr(X \cap Y) \geq \Pr(X) + \Pr(Y) - 1 \geq 1 - 2\alpha^{-1} - \beta^{-1}$.
This proves the statements on the output set.
The complexity statement follows from Theorem \ref{thm:ESP-cost}.

Towards bounding $\Pr(X)$ we combine Lemma \ref{lem:overlap} with Lemma \ref{lem:diff-core}.
Since $v \in \S'$, we know that
\[
\Pr(\tau_v(\S_d^c) \leq (120 \phi(\S))^{-1})
\leq 1/10.
\]
Our assumption on $\phi(\S)$ implies that $T \leq (120\phi(\S))^{-1}$, and so by Lemma \ref{lem:overlap} we get that
\[
\Pr\big( {\textstyle \min_{t\leq T}}\; d(\S_t \cap \S_d)/d(\S_t)
\geq 1 - \beta/10 \big)
\geq 1 - \beta^{-1}.
\]
Since $\tau \leq T$, the left-hand side lower bounds $\Pr(X)$, so that $\Pr(X) \geq 1 - \beta^{-1}$.

Towards bounding $\Pr(Y)$, we condition it on the possible stopping rule outcomes:
\begin{align*}
\Pr(Y)
&= \Pr(Y | \phi(\S_\tau) \leq \theta) \, \Pr(\phi(\S_\tau) \leq \theta)
+ \Pr(Y | \cost_\tau > B) \, \Pr(\cost_\tau > B) \\
&\quad + \Pr(Y | \tau = T) \, \Pr(\tau = T).
\end{align*}
Since $\theta = \gamma$ we know that $\Pr(Y | \phi(\S_\tau) \leq \theta) = 1$.
Next we wish to bound the second term by using Lemma \ref{lem:size}.
In combination with Markov's inequality this states that $\Pr(Z_\alpha) \leq 1/\alpha$, where we let $Z_\alpha$ denote the event that
\[
\cost_\tau/d(\S_\tau)
\geq \alpha(1+4\sqrt{T\log m})
= \alpha + 4B/(5M).
\]
If we lower bound $\Pr(Z_\alpha) \geq \Pr(Z_\alpha | \cost_\tau > B)\, \Pr(\cost_\tau > B)$, then we find $
\Pr(Z_\alpha | \cost_\tau > B) \, \Pr(\cost_\tau > B) \leq \alpha^{-1}$.
With $\bar{Z}_\alpha$ the negation of $Z_\alpha$, this leads to the bound
\[
\Pr(\bar{Z}_\alpha | \cost_\tau > B) \, \Pr(\cost_\tau > B)
\geq \Pr(\cost_\tau > B) - \alpha^{-1}.
\]
Now if both $\cost_\tau>B$ and $\bar{Z}_\alpha$ hold, then
\[
d(\S_\tau)
> \frac{\cost_\tau}{\alpha + 4B/(5M)}
> \frac{B}{\alpha + 4B/(5M)}
> M,
\]
using the bound $B > 5\alpha M$.
As a consequence, if both $\cost_\tau>B$ and $Z_\alpha$ hold, then also $Y$ holds, and so $\Pr(Y | \cost_\tau > B) \geq \Pr(Z_\alpha \ \cost_\tau > B)$.
This gives the desired bound
\begin{equation} \label{eq:PZB-bnd}
\Pr(\cost_\tau > B) \, \Pr(Y | \cost_\tau > B)
\geq \Pr(\cost_\tau > B) - \alpha^{-1}.
\end{equation}
Finally we will upper bound $\Pr(\tau = T)$ using Lemma \ref{lem:cond}.
For $c = \alpha$ this states that
\[
\Pr\big({\textstyle\min_{t<T}}\; \phi(\S_t) \leq \gamma\big)
\geq 1 - \alpha^{-1}.
\]
Since $\min_{t<T} \phi(\S_t) \leq \gamma$ implies that $\tau < T$, this shows that
\begin{equation} \label{eq:tauT-bnd}
\Pr(\tau = T)
\leq 1 - \Pr\big({\textstyle\min_{t<T}}\; \phi(\S_t) \leq \gamma\big)
\leq \alpha^{-1}.
\end{equation}
If now we combine the bounds \eqref{eq:PZB-bnd} and \eqref{eq:tauT-bnd} we get the final bound on $\Pr(Y)$:
\[
\Pr(Y)
\overset{\eqref{eq:PZB-bnd}}{\geq} \Pr(\phi(\S_\tau) \leq \theta)
	+ \Pr(\cost_\tau > B) - \alpha^{-1}
\geq 1 - \Pr(\tau = T) - \alpha^{-1}
\overset{\eqref{eq:tauT-bnd}}{\geq} 1 - 2\alpha^{-1}. \qedhere
\]
\end{proof}

\end{document}